\documentclass[12pt]{iopart}
\usepackage{amsthm}
\usepackage{amssymb}
\usepackage{graphicx}

\newtheorem{lemma}{Lemma}

\theoremstyle{definition}
\newtheorem{definition}{Definition}

\newcommand{\bra}[1]{\langle #1|}
\newcommand{\ket}[1]{| #1 \rangle }

\renewcommand{\tr}[1]{{\rm tr}[#1]}
\newcommand{\eqref}[1]{{[\ref{#1}]}}

\newcommand{\be}{\begin{eqnarray}}
\newcommand{\ee}{\end{eqnarray}}

\newcommand{\cE}{{\cal E}}

\newcommand{\cI}{{\cal I}}

\newcommand{\cF}{{\cal F}}

\newcommand{\cC}{{\cal C}}
\newcommand{\cS}{{\cal S}}
\newcommand{\cH}{{\cal H}}

\newcommand{\cL}{{\cal L}}

\newcommand{\cJ}{{\cal J}}

\newcommand{\sT}{{\sf T}}

\begin{document}
\title{Entanglement annihilating and entanglement breaking channels}
\author{Lenka Morav\v c\'\i kov\'a$^{1}$, M\'ario Ziman$^{1,2}$}
\address{
$^{1}$Research Center for Quantum Information, Institute of Physics, \\ Slovak Academy of Sciences, D\'ubravsk\'a cesta 9, 845 11 Bratislava, Slovakia \\
$^{2}$Faculty of Informatics, Masaryk University, \\ Botanick\'a 68a, 602 00 Brno, Czech Republic
}
\ead{ziman@savba.sk}
\begin{abstract}
We introduce and 
investigate a family of entanglement-annihilating channels. These channels
are capable to destroy any quantum entanglement within the system they act on. 
We show that they are not necessarily entanglement-breaking.
In order to achieve this result we analyze the subset of locally 
entanglement-annihilating channels. In this case, same 
local noise applied on each 
subsystem individually is less entanglement-annihilating 
(with respect to multi-partite entanglement) 
as the number of subsystems is increasing. Therefore, bipartite case
provides restrictions on the set of local entanglement-annihilating
channels for multipartite case. The introduced concepts are 
illustrated on the family of single-qubit depolarizing channels. 
\end{abstract}
\submitto{\JPA}

\section{Introduction}
The phenomenon of quantum entanglement \cite{schrodinger}
was recognized as an important resource in many applications of 
quantum information theory \cite{nielsen2000}. The parallel 
power of quantum computers as well as the security of quantum 
cryptosystems relies on the peculiar properties of entangled states of 
composite systems. For example, Shor's algorithm \cite{shor}, or
quantum teleportation \cite{teleportation}, could not be invented 
and successful without the puzzling properties of quantum 
entanglement. 

By definition, {\it entanglement} is a property 
assigned to multipartite quantum states. Following
Werner \cite{werner1984}, we say that a state 
$\omega$ of some bipartite system is separable, if it can be 
expressed as a convex combination of factorized states, i.e. 
written in the form
$\omega=\sum_j p_j \xi_j\otimes\zeta_j$. If it cannot, we say it 
is {\it entangled}. Entangled states exhibit their nonlocal origin 
in the following sense. Spatially separated experimentalists cannot 
create entanglement without some exchange of quantum systems, i.e. 
only by local actions and classical communication. 

The concepts of entanglement and separability 
can be directly generalized to the multipartite
case. Moreover, in this case a more subtle ``entanglement-induced"
separation of the state space is possible. For example, the so-called GHZ state 
$\ket{{\rm GHZ}}=\frac{1}{\sqrt{2}}(\ket{000}+\ket{111})$ is an example
of an entangled three-partite state, but no pair of subsystems 
is mutually entangled, because each pair is described by 
the classically correlated state 
$\frac{1}{2}(\ket{00}\bra{00}+\ket{11}\bra{11})$. Qualitatively different
family \cite{dur} 
of three-partite states is represented by
a state $\ket{W}=\frac{1}{\sqrt{3}}(\ket{001}+\ket{010}+\ket{100})$.
In this case, each pair of subsystems is entangled. This 
illustrates that the entanglement theory of multipartite systems 
is more complex 
and represents an interesting field of research. 

The tasks related to detection and characterization of entanglement 
represent prominent problems of quantum entanglement theory 
\cite{horodecki,plenio}. In this paper, we will pay attention to
entanglement dynamics induced by the evolution of individual
quantum subsystems. Since perfectly isolated quantum systems are very 
difficult to achieve in real experiments, unavoidable noise affects 
the states and potentially causes changes in the shared entanglement. 
It is of importance to understand the robustness of entanglement
with respect to these (local) processes.

For example, 
in order to perform the GHZ experiment \cite{ghz} the three-partite 
entangled state $\ket{{\rm GHZ}}$ must be distributed to the laboratories 
of Gina, Helen and Zoe. However, it is very likely that
the transmissions over long distances are not perfect and 
Gina, Helen and Zoe will actually receive and work with systems 
described by a modified quantum state $\omega^\prime_{{\rm GHZ}}$. 
Also, the storage of quantum systems and performing the 
experiments themselves represent additional sources of noise. Depending on the
particular type of overall noise the modified state may or may not 
be still used to perform the intended experiment and
observe a certain phenomenon. 

The better we understand the dynamical 
robustness of entanglement, the better we can perform 
multipartite experiments. Entanglement 
is the quantum resource, but it seems to be very
fragile, which reduces our abilities to conduct scalable (in any sense) 
experiments. For this purposes its distribution, storage and careful 
local manipulation is important. Loosely speaking,
our goal is to separate the "bad" noise from noise that is relatively 
"nice". In other words, which types of environmental influences 
must Gina, Helen and Zoe try to avoid, and which ones are acceptable? 
Intriguing questions are related to experiments with an increasing number of 
parties. Namely, are there local channels destroying any entanglement 
completely? How does it depend on the number of parties? For a given local 
channel is there always a number of parties, for which its action on each 
individual subsystem completely destroys any shared entanglement? Under the
presence of arbitrary local noise is there any limit on the number 
of particles that can be entangled? Such questions are partially addressed
in this paper.

The interplay between the entanglement we created
and its dynamical stability 
under particular sources of noise is currently a vivid field of research. 
Zanardi et al. in \cite{zanardi2000,zanardi2001} asked the question 
which unitary  transformations are better in creating entanglement. 
Linden et al. have shown in \cite{linden2009} that capacities of unitary 
channels to create and destroy entanglement are not the same. In particular, 
there are unitary channels that can create (on average) more entanglement than 
they can destroy. Zyczkowski et al. \cite{zyczkowski2001} 
analyzed the dynamics of entanglement for various different models 
of nonunitary evolutions. They showed the basic qualitative features
how entanglement evolves in time. Since then the phenomenon 
known as entanglement sudden death attracted relatively many 
researchers (see \cite{yu2009} and references therein) who
analyzed many dynamical models and made many 
observations concerning the entanglement dynamics. 
In \cite{ziman05,kinoshita,ziman06} it was shown that
local channels do not preserve entanglement-induced ordering. 
After its action originally more entangled states can become 
less entangled than states coming from some
originally less entangled states. Recently, an evolution equation 
(in fact inequality) for the entanglement affected by local independent
noise has been formulated \cite{tiersch2008,yu2008,zg_li2009}.

In this paper we focus on 
channels that completely destroy any entanglement. Clearly,
unitary channels do not possess such property. For them entanglement
creation goes always in hand with entanglement annihilation. However,
the situation becomes more interesting when general nonunitary evolutions
are considered. In Section II we present our definitions and list the 
basic properties 
of the so-called {\it entanglement-breaking} and 
{\it entanglement-annihilating} channels. 
From the perspective of these concepts, in Section III we 
investigate the channels acting locally on a multipartite composite system. 
In Section IV we study a family of depolarizing channels in more detail.
Finally, we summarize our observations in Section V.

\section{Preliminaries}
A composite quantum system $Q$ consisting of $n$ quantum systems 
is associated with a Hilbert space 
$\cH_Q\equiv\cH^{(n)}=\cH_1\otimes\cdots\otimes\cH_n$. 
Its states are represented by so-called density operators, i.e.
positive operators with a unit trace. Let us denote by
$\cS(\cH)=\{\varrho:\varrho\geq O, \tr{\varrho}=1\}$ the set
of all states of a system associated with the Hilbert space $\cH$.
We divide the system $Q$ into two subsystems $A$ and $B$ consisting
of $k$ and $n-k$ particles with Hilbert spaces $\cH_A=\cH^{(k)}$, 
$\cH_B=\cH^{(n\setminus k)}$, respectively. When denoting 
the total Hilbert space as $\cH_{AB}=\cH_A\otimes\cH_B$
 we mean that the whole system is understood as
a bipartite system consisting of subsystems $A$ and $B$.

Any Hilbert space $\cH$ with a defined tensor structure 
we can divide into two subsets $\cS_{\rm ent}(\cH),\cS_{\rm sep}(\cH)$ 
of entangled and separable states. In particular,
$\cS_{\rm sep}(\cH_Q)$ is the set of all separable states with respect
to the division of $\cH_Q$ into $n$ particles 
($n$-partite separability), i.e. it consists of states of the form
$\varrho=\sum_j p_j \varrho_1^{(j)}\otimes\cdots\otimes\varrho_n^{(j)}$. 
Similarly, the set $\cS_{\rm sep}(\cH_A)$
contains separable states of $k$ particles forming the subsystem
$A$. However, we will use $\cS_{\rm sep}(\cH_{AB})$ to denote the
set of separable states with respect to division of $\cH_{AB}$ 
into subsystems $A$ and $B$ (bipartite separability), i.e.
this set consists of states that can be expressed as 
$\varrho=\sum_j p_j\varrho_A^{(j)}\otimes\varrho_B^{(j)}$.
In such case the internal structure of the 
composite subsystems $A$ and $B$ is irrelevant and 
$\cS_{\rm sep}(\cH_Q)\varsubsetneq\cS_{\rm sep}(\cH_{AB})$ (meaning
$\cS_{\rm sep}(\cH_Q)$ is a strict subset of $\cS_{\rm sep}(\cH_{AB})$).
The analogous notation will be used for subsets of entangled states 
$\cS_{\rm ent}(\cH_Q)$, $\cS_{\rm ent}(\cH_A)$, and $\cS_{\rm ent}(\cH_{AB})$.

The evolution of quantum systems is described by means of quantum 
channels, i.e. completely positive trace-preserving linear maps $\cE$ 
defined on the set of all linear operators $\cL(\cH)$ on the considered 
Hilbert space $\cH$. A linear mapping $\cE:\cL(\cH)\to\cL(\cH)$
defines a quantum channel if $\tr{\cE[X]}=\tr{X}$ for all $X\in\cL(\cH)$
and $(\cE\otimes\cI)[X]$ remains a positive operator for all positive operators
$X\in\cL(\cH\otimes\cH_{\rm anc})$, where $\cH_{\rm anc}$ is the Hilbert space
associated with an ancillary system of arbitrary size. We use 
$\cI$ to denote the identity (trivial) channel on the ancillary system. 
In what follows we will denote the ancillary system by $B$, thus, 
in our further consideration the subsystem $B$ can be of arbitrary size 
and structure.

\begin{definition}
We say the channel $\cE_A$ acting on the subsystem $\cH_A$ is
\begin{itemize}
\item{\it entanglement-annihilating} (EA) if
$$\cE_A[\cS(\cH_A]\subset\cS_{\rm sep}(\cH_A)\,.$$
\item{\it entanglement-breaking} (EB) if 
$$\cE_A\otimes\cI_B[\cS(\cH_{AB})]\subset\cS_{\rm sep}(\cH_{AB})$$
for arbitrary ancillary system $B$.
\end{itemize}
\end{definition}
Thus, the entanglement-annihilating (EA) channels are defined as the ones
that completely destroy/annihilate any entanglement within the 
subset $A$ of the composite system (see Fig.~\ref{fig:1}). 
On contrary, the entanglement-breaking 
(EB) channels are those that completely disentangle the subsystem 
they are acting on from the rest of the system. Note that 
by definition the EA channels (acting on subsystem $A$) do not 
necessarily disentangle the subsystems $A$ and $B$. Similarly, the
EB channels do not necessarily destroy entanglement within the subsystem $A$. 
The two concepts are thus (by definition) different and our aim 
is to investigate their mutual relationship.

\begin{figure}
\begin{center}
\includegraphics[width=8cm]{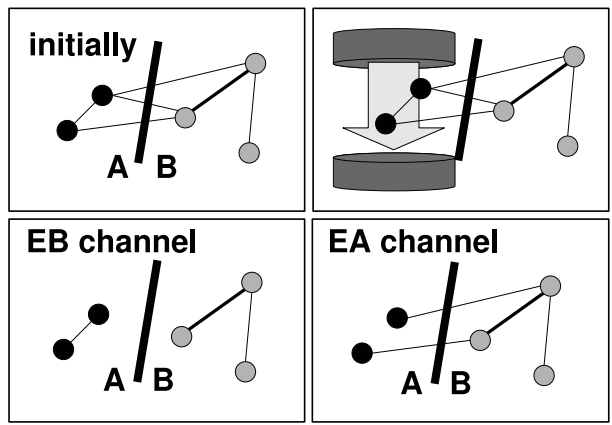}
\caption{The action of entanglement-breaking and entanglement-annihilating
channel is illustrated. Lines between the systems exhibits the existence of 
entanglement.}
\end{center}
\label{fig:1}
\end{figure}

Let us denote by $\sT(\cH)$ the set of all linear maps 
$\cE:\cL(\cH)\to\cL(\cH)$. Let $\sT_{\rm chan}(\cH)$ be
the set of all channels on system associated
with a Hilbert space $\cH$, i.e. $\sT_{\rm chan}(\cH_A)$ 
are all channels defined on
subsystem $A$. Let $\sT_{\rm EA}(\cH_A)$ and 
$\sT_{\rm EB}(\cH_A)$ denote the subsets of EA and EB channels, respectively.

\subsection{Basic properties}

In quantum  theory,  measurements are associated with so-called
positive operators valued measures (POVMs), i.e. collections of positive 
operators $F_1,\dots,F_n$ such that $\sum_j F_j=I$. As it was shown
in Ref.~\cite{horodecki_eb} any entanglement-breaking
channel can be understood as a {\it measure and prepare} procedure.
That is, each EB channel can be expressed in the form
\be
\cE_A[\,\cdot\,]=\sum_j \tr{\,\cdot\,F_j}\,\varrho_j\,,
\label{eq:eb_channel}
\ee
for some POVM $\{F_j\}$ and some fixed states
$\varrho_1,\dots,\varrho_n$. 

To verify whether a given channel is entanglement-breaking or not, 
is equivalent to detecting whether a specific bipartite quantum state 
is separable, or entangled. Let us denote by  $P_+=\ket{\psi_+}\bra{\psi_+}$ 
a projector onto the maximally entangled vector state 
$\ket{\psi_+}=\frac{1}{\sqrt{d}}\sum_j \ket{\varphi_j\otimes\varphi_j}$
of Hilbert space $\cH_A\otimes\cH_B$, where $\cH_B\equiv\cH_A$,
and vectors $\{\ket{\varphi_j}\}_j$ form an orthonormal basis of the 
Hilbert space $\cH_A$ of dimension $d$. A mapping 
$\cJ:\sT(\cH_A)\to\cL(\cH_A\otimes\cH_B)$ defined via the identity 
\cite{pillis,choi,jamiolkowski}
\be
\cJ(\cE_A)=(\cE_A\otimes\cI_B)[P_+]\equiv\Omega_\cE\,,
\ee
determines a unique operator $\Omega_\cE$ for each linear mapping
$\cE\in\sT(\cH_A)$. It is known as the Choi-Jamiolkowski
isomorphism. The Choi-Jamiolkowski operator
$\Omega_\cE$ provides an alternative representation of a quantum channel
(acting on the system associated with the Hilbert space $\cH_A$)
as a specific linear operator on the Hilbert space $\cH_A\otimes\cH_B$.
The complete positivity of $\cE_A$ is translated to the
positivity of $\Omega_\cE$ and the trace-preserving condition is equivalent
with ${\rm tr}_{A}\Omega_\cE=\frac{1}{d}I$, thus $\Omega_\cE$
is a valid density operator on $\cH_A\otimes\cH_B$. Clearly,
if $\cE_A$ is entanglement-breaking, then $\Omega_\cE$ is a 
separable state on $\cH_A\otimes\cH_B$. Surprisingly, the inverse 
implication is also true \cite{horodecki_eb}, 
i.e. the separability of $\Omega_\cE$ is necessary and sufficient
for $\cE_A$ being entanglement-breaking. This significantly simplifies the 
analysis of channels with respect to entanglement-breaking, because
it is sufficient to test the action of the channel only on a single
state -- the maximally entangled state and test whether $\Omega_\cE$
is separable, or not. 

While entanglement-breaking channels had been already investigated, 
the concept of entanglement-annihilating channels is new. An interesting
question is how to test whether a given channel is 
entanglement-annihilating, or not. If $\cE[P_\psi]$ with
$P_\psi=\ket{\psi}\bra{\psi}$ is separable for all pure states
$\ket{\psi}\in\cH$, then for any state $\omega\in\cS(\cH)$ the state
$\cE[\omega]$ is separable as well. This follows from
the fact that the set of separable states is convex and and
any state $\omega$ can be decomposed into a convex 
combination of pure states, i.e. $\omega=\sum_j p_j P_{\psi_j}$.
That is, whether, or not the channel is entanglement-annihilating, 
it is sufficient to test its action only on pure states. 
Unfortunately, this is still not an easy task. We left open whether 
there exists a simpler way of testing for the EA property. 

Let us continue with simple observations on elementary properties of
the sets of entanglement-annihilating and entanglement-breaking channels.
\begin{lemma}
$\sT_{\rm EA}(\cH_A),\sT_{\rm EB}(\cH_A)$ are convex.
\end{lemma}
\begin{proof}
By definition, if $\cE_1,\cE_2\in\sT_{\rm EA}$, then
$\cE_j[\cS(\cH_A)]\subset\cS_{\rm sep}(\cH_A)$ for $j=1,2$.
Due to convexity of $\cE_j(\cS(\cH_A))$ and $\cS_{\rm sep}(\cH_A)$
it follows that also $(\lambda\cE_1+(1-\lambda)\cE_2)[\cS(\cH_A)]
\subset\cS_{\rm sep}(\cH_A)$. Similarly for the case of EB channels.
\end{proof}

\begin{lemma}
If $\cE\in\sT_{\rm EA}(\cH_A)$ and $\cF\in\sT(\cH_A)$,
then $\cE\cdot\cF\in\sT_{\rm EA}(\cH_A)$.
\end{lemma}
\begin{proof}
Defining property of $\cE$ implies that
$\cE[\cF[\cS(\cH_A)]]\subset\cS_{\rm sep}(\cH_A)$, hence
$\cE\cdot\cF$ is an entanglement-annihilating channel. 
\end{proof}

\begin{lemma}
If $\cE\in\sT_{\rm EB}(\cH_A)$ and $\cF\in\sT(\cH_A)$, then 
$\cE\cdot\cF$, $\cF\cdot\cE\in\sT_{\rm EB}(\cH_A)$.
\end{lemma}
\begin{proof}
Since $\cE\in\sT_{\rm EB}(\cH_A)$ it follows that 
$(\cE\otimes\cI)[(\cF\otimes\cI)[\cS(\cH_{AB})]]\subset\cS_{\rm sep}(\cH_{AB})$.
For any $\cF\in\sT(\cH_A)$ the channel $\cF\otimes\cI$ cannot
create entanglement (out of separable state) 
between subsystems $A$ and $B$. Therefore,
$(\cF\otimes\cI)[(\cE\otimes\cI)[\cS(\cH_{AB})]]\subset
\cS_{\rm sep}(\cH_{AB})$, hence the both channels $\cE\cdot\cF,\cF\cdot\cE$
are entanglement breaking providing that one of them is.
\end{proof}

In what follows we will investigate the set relation between 
$\sT_{\rm EA}\equiv\sT_{\rm EA}(\cH_A)$ 
and $\sT_{\rm EB}\equiv\sT_{\rm EB}(\cH_A)$ (see Fig.~\ref{fig:2}), 
both defined on the same Hilbert space
$\cH_A$. As a consequence of the above lemmas we get 
that a composition $\cE\cdot\cF$
of the entanglement-breaking channel $\cF$ and of the entanglement-annihilating 
channel $\cE$ belongs to the intersection 
$\sT_{\rm EA}\cap\sT_{\rm EB}$. 
That is, there are channels which are simultaneously EB and EA.
On the other hand, although the channel $\cF\cdot\cE$ is necessarily 
entanglement-breaking, it does not have to be entanglement-annihilating. 
For example, a single-point contraction $\cF$ of the whole state space
into a single entangled state $\omega\in\cS_{\rm ent}(\cH_A)$ 
is entanglement-breaking, because it can be expressed in the form 
$\cF[\cdot]=\sum_j \tr{\cdot F_j}\omega=\omega$. However, the 
channel $\cF\cdot\cE$ is not entanglement-annihilating 
for arbitrary $\cE\in\sT_{\rm EA}$, because $(\cF\cdot\cE)[\cS(\cH_A)]
=\omega\in\cS_{\rm ent}(\cH_A)$. This means there are entanglement-breaking 
channels which are not entanglement-annihilating, i.e.
$\sT_{\rm EB}\not\subset\sT_{\rm EA}$. Later on we shall get back 
also to the inverse question whether $\sT_{\rm EA}\subset\sT_{\rm EB}$, 
or not.

\begin{lemma}
Let $\cE[\cdot]=\sum_j \tr{\cdot F_j}\varrho_j$, i.e. $\cE\in\sT_{\rm EB}$. 
Then the following statements hold:
\begin{itemize}
\item[(i)]
If $\varrho_j\in\cS_{\rm sep}(\cH_A)$ for all $j$, then
$\cE\in\sT_{\rm EA}$.
\item[(ii)]
If there exists $\ket{\varphi}\in\cH_A$ such that 
$F_j\ket{\varphi}=\ket{\varphi}$ for some $j$, 
then $\cE\in\sT_{\rm EA}$ only if $\varrho_j$ is separable.
\end{itemize}
\end{lemma}
\begin{proof}
The first part {\it (i)} is obvious, because convex sum of separable states
is necessarily a separable state. The second
part {\it (ii)} follows from the formula
$\cE[\ket{\varphi}\bra{\varphi}]=\varrho_j$, which implies that 
being entanglement-annihilating requires that $\varrho_j$ is separable. 
\end{proof}
The second half of this lemma can be used to show that its first half
can not be {\it if and only if} statement. Consider an entangled state $\omega$.
Let us define $\kappa$ as the largest value of $x\in[0,1]$ 
for which the state $x\omega+(1-x)\frac{1}{d}I$ is separable.
This value is strictly larger than 0.
Let $F$ be a positive operator with all the eigenvalues smaller
than $\kappa$. Then
$\cE[\cdot]=\tr{\cdot F}\omega+\tr{\cdot(I-F)}\frac{1}{d}I$
defines an entanglement-annihilating channel, because
$\cE[\varrho]=x\omega+(1-x)\frac{1}{d}I$ with $x=\tr{\varrho F}<\kappa$.
Thus, $\cE\in\sT_{\rm EA}\cap\sT_{\rm EB}$ does not imply that
all $\varrho_j$ are necessarily separable.

\begin{figure}
\begin{center}
\includegraphics[width=8cm]{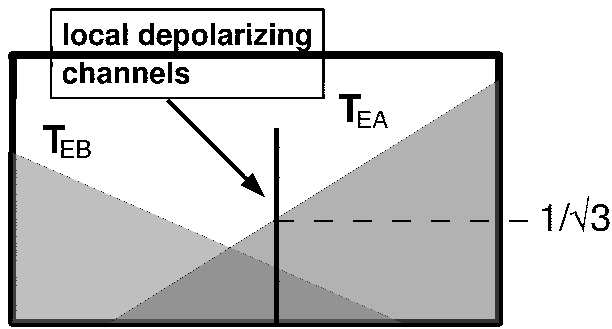}
\caption{This figure schematically illustrates the subsets 
of entanglement-annihilating and
entanglement-breaking channels. The family of two-qubit local 
depolarizing channels $\cE_\lambda\otimes\cE_\lambda$ is depicted, too. }
\label{fig:2}
\end{center}
\end{figure}

\section{Local channels}
We distinguish two basic types of channels acting on a composite system
of $k$ particles: global and local. We say a channel $\cF$ is local 
if it has a tensor product form 
$\cF=\cE_1\otimes\cdots\otimes\cE_k$, where
$\cE_j$ are channels acting on individual particles. If a channel
does not have the factorized form we say it is global. In what follows 
we will investigate entanglement dynamics under local channels. 
Moreover, we will assume that Hilbert spaces of all particles are isomorphic
and each particle undergoes the same evolution, i.e. $\cE_j=\cE$ for all $j$. 
Although this is not the most general case, under certain circumstances 
it is of physical relevance. 

We say a single-particle channel $\cE$ is a {\it $k$-locally entanglement
annihilating} channel ($k$-LEA), 
if $\cE^{\otimes k}\in\sT_{\rm EA}(\cH^{\otimes k})$.
Similarly, $\cE$ is a {\it $k$-locally entanglement breaking} channel ($k$-LEB) 
if $\cE^{\otimes k}\in\sT_{\rm EB}(\cH^{\otimes k})$.
By $\sT_{k\rm -LEA}$, $\sT_{k\rm -LEB}$ we shall denote the
subsets of $k$-LEA and $k$-LEB channels, respectively. 
Since elements of these sets are uniquely associated
with single particle channels $\cE\in\sT_{\rm chan}(\cH)$, 
we can understand these sets as subsets of $\sT_{\rm chan}(\cH)$, i.e. 
$\sT_{k\rm -LEA},\sT_{k\rm -LEB}\subset\sT_{\rm chan}(\cH)$. 
Moreover, let us denote by $\sT^1_{\rm EB}\subset\sT_{\rm chann}(\cH)$ 
the subset of entanglement breaking channels acting on the single system $\cH$,
i.e. $\cE\in\sT^1_{\rm EB}$ means that $(\cE\otimes\cI_{\rm anc})[\omega]$ 
is separable for all $\omega\in\cS(\cH\otimes\cH_{\rm anc})$. 
Our goal is to analyze 
the relation between the subsets $\sT_{k\rm -LEA}$,
$\sT_{k\rm -LEB}$ and $\sT^1_{\rm EB}$. The ultimate question is which 
single particle channels $\cE$ if applied to a suitable
number of particles, necessarily destroy any entanglement
within the $k$-partite system. 

By definition single-particle
entanglement-breaking channels disentangle each particle
from the rest of the system. Therefore, they are 
simultaneously locally entanglement-annihilating and 
locally entanglement-breaking channels for any value of $k$, hence
\be
\sT^1_{\rm EB}\subset\sT_{k\rm -LEB}\,;\quad
\sT^1_{\rm EB}\subset\sT_{k\rm -LEA}\,.
\ee
Further, consider a channel $\cE\in\sT_{k\rm -LEA}$. It means $\cE^{\otimes k}$
transforms any state of $k$ particles into some separable state of
$k$ particles, i.e. $\cE^{\otimes k}[\omega]=\sum_j p_j \xi_j^{(1)}\otimes\cdots
\otimes\xi_j^{(k)}$. Setting $\omega=\varrho_0\otimes\omega^\prime$
we get that $\cE^{\otimes(k-1)}[\omega^\prime]$ is separable for any
state $\omega^\prime\in\cS(\cH^{\otimes(k-1)})$. That is, $\cE$
is also $(k-1)$-LEA channel. Consequently, we can write the relation
\be
\sT_{k\rm -LEA}\subset\sT_{l\rm -LEA} \quad {\rm for}\quad k>l\, ,
\ee
which implies
\be
\label{eq:lea_rel}
\sT_{\rm EB}^1\subset
\sT_{\rm\infty-LEA}\subset\cdots\subset\sT_{\rm 3-LEA}\subset\sT_{\rm 2-LEA}\,.
\ee

On the other hand for a channel $\cE^{\otimes k}$ the corresponding
Choi operator takes the form $\Omega_\cE^{\otimes k}$, where
$\Omega_\cE=(\cE\otimes\cI)[P_+]$. If $\cE^{\otimes k}$ is
entanglement breaking, i.e. $\cE\in\sT_{k\rm -LEB}$, then
$\Omega_\cE^{\otimes k}$ is separable with respect to a bipartite
splitting into $k$ principal systems and $k$ ancillary systems. However,
this implies that $\Omega_\cE$ itself is separable, hence, the
single particle channel $\cE$ is entanglement breaking, 
i.e. $\cE\in\sT^1_{\rm EB}$. As a result we get the following
set identities
\be
\sT_{\rm EB}^1 =
\sT_{\rm 2-LEB}=\sT_{\rm 3-LEB}=\cdots =\sT_{\rm \infty-LEB}\,.
\ee 
\section{A case study: depolarizing channels}
In this section we will address the question whether the entanglement-breaking
channels are not the only locally entanglement-annihilating channels. We will
give explicit example of qubit channels that are not entanglement-breaking,
but completely destroy entanglement if applied on individual
particles.

Consider a one-parametric family of {\it depolarizing channels}
\be
\cE_\lambda[X]=\lambda X+(1-\lambda)\tr{X}\,\frac{1}{d}I\,,
\ee
where $\lambda\in[0,1]$. Applying this channel to the maximally
entangled state $P_+$ we get the so-called Werner states \cite{werner1984}
\be
\Omega_\lambda=\lambda P_++(1-\lambda)\frac{1}{d}I\otimes \frac{1}{d}I\,.
\ee
For qubit ($d=2$) the states $\Omega_\lambda$ are separable 
for $\lambda\leq 1/3$. Therefore, if $\lambda\leq 1/3$ the qubit
depolarizing channel $\cE_\lambda$ is entanglement-breaking.

\subsection{2-LEA channels}
The local channel $\cE_\lambda\otimes\cE_\lambda$ acts as follows
\be
\nonumber
\omega_{12}^\prime&=&(\cE_\lambda\otimes\cE_\lambda)[\omega_{12}]
=\lambda^2\omega_{12}+(1-\lambda)^2\frac{1}{d}I\otimes \frac{1}{d}I\\
& & +\lambda(1-\lambda)(\omega_1\otimes \frac{1}{d}I+
\frac{1}{d}I\otimes\omega_2)\,,
\ee
where $\omega_1={\rm tr}_2[\omega_{12}]$ and $\omega_2={\rm tr}_1[\omega_{12}]$.
Unlike in the analysis of entanglement-breaking channels we need
to verify the separability for all input states $\omega_{12}$ in order
to conclude that the channel $\cE_\lambda$ is $2$-LEA. Fortunately,
it is sufficient to analyze the separability for pure states
$\omega_{12}=\ket{\psi}\bra{\psi}$ only, because the set of separable
states is convex and channels preserve the convexity. Let us note that
$\cE_\lambda=\lambda\cI+(1-\lambda)\cC_0$, where $\cC_0$ denotes
the contraction of the whole state space into the complete mixture state 
$\frac{1}{d}I$. Both channels, $\cI,\cC_0$ commutes with unitary
transformations, i.e. $\cI[UXU^\dagger]=U\cI[X] U^\dagger$
and $\cC_0[UXU^\dagger]=U\cC_0[X] U^\dagger$. Consequently, 
$\cE_\lambda\otimes\cE_\lambda$ commutes with all unitary channels 
$U\otimes V$. Using such local unitary channels any
vector $\ket{\psi}$ can be written in its Schmidt form
\be
\ket{\psi}=\sum_j \sqrt{q_j}\ket{\varphi_j\otimes\varphi_j^\prime}\,,
\ee
where $\{\ket{\varphi_j}\}$, $\{\ket{\varphi_j^\prime}\}$ are suitable 
orthonormal bases of the first and the second particle, respectively.
Because of the unitary invariance of $\cE_\lambda\otimes\cE_\lambda$
it is sufficient to consider only vectors expressed in
a single fixed Schmidt basis. The reduced states take the 
diagonal form
\be
\omega_1=\sum_j q_j\ket{\varphi_j}\bra{\varphi_j}\,,\quad
\omega_2=\sum_j q_j\ket{\varphi_j^\prime}\bra{\varphi_j^\prime}\,.
\ee

Further, let us analyze the case of qubit, i.e. $d=2$ and
$\ket{\psi}=\sqrt{q_0}\ket{\varphi_0\otimes\varphi_0^\prime}+
\sqrt{q_1}\ket{\varphi_1\otimes\varphi_1^\prime}$, thus,
\begin{eqnarray}
\nonumber
\omega_{12}^\prime=
\frac{1}{4}\left(
\begin{array}{cccc}
\lambda_-^2+4\lambda q_0 & 0 & 0 & 4\lambda^2\,\sqrt{q_0 q_1}\\
0 & \lambda_+\lambda_- & 0 & 0 \\
0 & 0 & \lambda_+\lambda_- & 0 \\
4\lambda^2\,\sqrt{q_0 q_1} & 0 & 0 & \lambda_-^2+4\lambda q_1
\end{array}
\right)\,,
\end{eqnarray}
where we set $\lambda_\pm=1\pm\lambda$. The eigenvalues of the partially 
transposed operator $\omega_{12}^{\prime\Gamma}$ reads
\be
\mu_1&=&\frac{1}{4}(1-\lambda)^2+\lambda q_0\;;\\
\mu_2&=&\frac{1}{4}(1-\lambda)^2+\lambda q_1\;;\\
\mu_\pm &=& 
\frac{1}{4}(1-\lambda^2\pm 4\lambda^2\sqrt{q_0 q_1})\;;
\ee
where $\lambda\in[0,1]$, $q_0\in[0,1]$ and $q_1=1-q_0$. According to 
Peres-Horodecki criterion \cite{peres,horodecki_transp}
a two-qubit state $\omega$ is separable if and only if $\omega^\Gamma$
is positive. For larger systems this separability
criterion is not sufficient to unambiguously 
distinguish between entangled and separable states. From nonpositivity
of $\omega^\Gamma$ we can conclude that the state is entangled, but
the inverse implication does not hold.

All the eigenvalues of $\omega_{12}^{\prime\Gamma}$ except $\mu_-$ 
are always positive, hence it is sufficient to analyze only
this one. In particular, if for a fixed value of
$\lambda$ the eigenvalue $\mu_-$ is positive for all values $q_0,q_1$, 
then the corresponding depolarizing channel $\cE_\lambda$ is 
a 2-locally entanglement-annihilating channel. Thus, we want to minimize
$\mu_-$ over the interval $q_0\in[0,1]$ for each depolarizing 
channel $\cE_\lambda$. Fortunately, for each $\lambda$ the 
minimum is achieved for the same value $q_0=q_1=1/2$. Consequently,
the partially transposed operator $\omega_{12}^{\prime\Gamma}$ 
is positive if and only if $1-3\lambda^2\ge 0$, i.e.
\be
\lambda\leq\frac{1}{\sqrt{3}}\approx 0.577\,.
\ee
In summary, the qubit depolarizing channel $\cE_\lambda$ is 2-locally 
entanglement-annihilating if and only if $\lambda\in[0,1/\sqrt{3}]$, 
whereas it is entanglement-breaking for $\lambda\in[0,1/3]$, 
hence $\sT_{\rm EB}^1\neq\sT_{2-\rm LEA}$.

\subsection{3-LEA channels}
We have already shown (see Eq.\eqref{eq:lea_rel}) that 
$\sT_{\rm 3-LEA}\subset\sT_{\rm 2-LEA}$. In this section we will investigate 
the inverse relation, namely whether 2-LEA depolarizing channels
$\cE_{\lambda}$ are necessarily also 3-LEA channels. Under the action
of $\cE_\lambda\otimes\cE_\lambda\otimes\cE_\lambda$
a general three-partite state $\omega_{123}$ is transformed
into the state
\be
\nonumber
\omega^\prime_{123}&=&\lambda^3\omega_{123}+
\frac{1}{d^3}(1-\lambda)^3 I_1\otimes I_2\otimes I_3\\
\nonumber & & +\frac{1}{d}\lambda^2(1-\lambda)
(\omega_{12}\otimes I_3+\omega_{13}\otimes I_2+\omega_{23}\otimes I_1)\\
\nonumber & & 
+\frac{1}{d^2}\lambda(1-\lambda)^2
(\omega_{1}\otimes I_{23}+\omega_{2}\otimes I_{13}+\omega_{3}\otimes I_{12})
\,,
\ee
where $I_{jk}=I_j\otimes I_k$ and $I_j$ stands for the identity operator 
on the $j$th particle. Since $\cE_\lambda\in\sT_{2-\rm LEA}$ the
reduced bipartite states $\omega_{12}^\prime,\omega_{13}^\prime,
\omega_{23}^\prime$ are separable. If some entanglement has left in the
composite system, then it must be visible with respect to 
bipartite partitionings $1|23$, or $2|13$, or $3|12$.

As before let us assume the case of qubits ($d=2$) and set 
$\omega_{123}=\ket{{\rm GHZ}}\bra{{\rm GHZ}}$, i.e.
$\omega_{12}=\omega_{13}=\omega_{12}=
\frac{1}{2}(\ket{00}\bra{00}+\ket{11}\bra{11})\equiv\Theta$
and $\omega_1=\omega_2=\omega_3=\frac{1}{2}I$. Thus,
\be
\nonumber
\omega_{\rm GHZ}^\prime&=&\lambda^3\ket{{\rm GHZ}}\bra{{\rm GHZ}}
+\frac{(1-\lambda)^2}{8}I\otimes I\otimes I\\
\nonumber & & +\frac{1}{2}\lambda^2(1-\lambda)
(\Theta_{12}\otimes I+\Theta_{13}\otimes I+\Theta_{23}\otimes I)
\ee
Let us consider the splitting $1|23$ and define the 
basis elements of system $23$ as follows
\be
\ket{{\bf 0}}=\ket{00}\,,
\ket{{\bf 1}}=\ket{11}\,,
\ket{{\bf 2}}=\ket{01}\,,
\ket{{\bf 3}}=\ket{10}\,.
\ee
In this basis
\be
\nonumber
\omega_{\rm GHZ}^\prime&=&
\frac{1}{8}(1-\lambda^2)
\ket{0}\bra{0}\otimes (I_{23}-\ket{{\bf 0}}\bra{{\bf 0}})\\
\nonumber & & 
+\frac{1}{8}(1-\lambda^2)
\ket{1}\bra{1}\otimes (I_{23}-\ket{{\bf 1}}\bra{{\bf 1}})]\\
\nonumber & &
+\frac{1}{8}(1+3\lambda^2)(\ket{0{\bf 0}}\bra{0{\bf 0}}+
\ket{1{\bf 1}}\bra{1{\bf 1}})\\
\nonumber & & 
+\frac{1}{2}\lambda^3(\ket{0{\bf 0}}\bra{1{\bf 1}}+
\ket{1{\bf 1}}\bra{0{\bf 0}})\,.
\ee
The last term plays the crucial role from the point
of partial transposition criterion applied
with respect to splitting $1|23$. Let us note that 
due to symmetry for different splitting $2|13$ and $3|12$ 
we will derive qualitatively the same bipartite density matrix. 
That is, if the state $\omega^\prime_{\rm GHZ}$ is entangled with
respect to the splitting $1|23$, then it is also entangled with respect
to remaining bipartite splittings.

Among all the eigenvalues 
of $\omega_{\rm GHZ}^{\prime\Gamma}$ only
\be
\mu_-=\frac{1}{2}(\frac{1}{4}(1-\lambda^2)-\lambda^3)\, ,
\ee
is negative when $\lambda>0.5567$. Important for us, is that 
for $\lambda=1/\sqrt{3}$ this eigenvalue is negative, hence the 
state remains entangled although the depolarizing channel is 2-LEA. 
Therefore, we can conclude that
\be
\sT_{3\rm -LEA}\subsetneq\sT_{2\rm -LEA}\,.
\ee
Since partial transposition criterion is not sufficient
to conclude the separability, it cannot be used to decide for 
which $\lambda$ the channel $\cE_\lambda$ is 3-LEA and for which it is not. 
If $\lambda > 0.5567$ we can safely say that $\cE_\lambda$ is not 
the 3-locally entanglement-annihilating channel. However, for smaller 
values we cannot exclude the possibility that the channel $\cE_\lambda$ 
does not belong to $\sT_{\rm 3-LEA}$ unless $\lambda\leq 1/3$, when the channel 
is entanglement-breaking.

\subsection{EA vs EB}
In this section we shall get back to the question on relation between
EA and EB channels. We have already shown that there are 
entanglement-breaking channels which are not entanglement-annihilating. 
In what follows we are interested in whether the opposite case is also 
possible, i.e. whether there are entanglement-annihilating channels 
that are entanglement-breaking. Mathematically, we are asking which 
of the relations
$\sT_{\rm EA}\subset\sT_{\rm EB}$, 
$\sT_{\rm EA}\not\subset\sT_{\rm EB}$ hold. 
The derived results we can use to argue whether 2-LEA channels 
$\cE_\lambda\otimes\cE_\lambda$ considered as a channels acting on bipartite 
system are entanglement-breaking, or not. Certainly, they are 
entanglement-annihilating because 
$\sT_{\rm 2-LEA}(\cH\otimes\cH)\subset\sT_{\rm EA}(\cH\otimes\cH)$.
We have seen that for the value $\lambda=1/\sqrt{3}$
the channel $\cE_{\lambda}\otimes\cE_{\lambda}$ is
entanglement-annihilating. More importantly, we have also shown that for
the same value $\cE_{\lambda}\otimes\cE_{\lambda}\otimes\cE_\lambda$ is not
entanglement-annihilating. Let us formally write
$\cE_{\lambda}\otimes\cE_{\lambda}\otimes\cE_\lambda=
(\cI\otimes\cI\otimes\cE_\lambda)(\cE_\lambda\otimes\cE_\lambda\otimes\cI)$.
Since channels of the form $\cI\otimes\cI\otimes\cE$ can only decrease the 
entanglement and since $\cE_{\lambda}^{\otimes 3}[\ket{\rm GHZ}\bra{\rm GHZ}]$
is entangled for $\lambda=1/\sqrt{3}$, we can conclude that
also $(\cE_\lambda\otimes\cE_\lambda)\otimes\cI [\ket{\rm GHZ}\bra{\rm GHZ}]$
is entangled. But this is in contradiction with the assumption
that $\cE_{\lambda}\otimes\cE_\lambda$ is entanglement-breaking.
Based on this example we can conclude that
\be
\sT_{\rm EA}\not\subset\sT_{\rm EB}\,,
\ee
that is, there are entanglement-annihilating channels which are
not entanglement-breaking.

\section{Summary}
In  this paper we introduced the concept of entanglement-annihilating  
channels as the channels that completely destroy any entanglement within
the systems they act on. We investigated the structural properties of 
the set of these channels and its relation to the set of entanglement-breaking
channels $\sT_{\rm EB}$, i.e. channels that completely destroy 
entanglement between the subsystem they act on and the 
rest of the composite system (see Fig.~\ref{fig:1}). 
In particular, we have shown that
\be
\sT_{\rm EA}\cap\sT_{\rm EB}\neq\emptyset\,,\\
\sT_{\rm EB}\not\subset\sT_{\rm EA}\not\subset\sT_{\rm EB}\,.
\ee
That is, there are channels which are simultaneously
entanglement-breaking and entanglement-annihilating, but
also channels possessing only one of this features.
The set of entanglement-annihilating channels $\sT_{\rm EA}$ is convex.
Moreover, a composition of an entanglement-annihilating channel and an arbitrary
channel results in an entanglement-annihilating channel, i.e. the property
of being entanglement-annihilating is preserved under channel 
composition. 

One of the above relations we were able to prove by analyzing
the family of local depolarizing channels. We defined 
the so-called $k$-local channels as channels 
of the form $\cE\otimes\cdots\otimes\cE$.
That is, the same noise $\cE$ is applied on each individual subsystem 
forming a composite $k$-partite system. We investigated when 
a single-particle channel $\cE$ constitutes
a $k$-locally entanglement-annihilating channel ($k$-LEA), or a $k$-locally
entanglement-breaking channel ($k$-LEB). In particular, for depolarizing 
qubit channels we found that for $\lambda\leq 1/\sqrt{3}$ the channel is 
$2$-locally entanglement annihilating, while for $\lambda>1/3$ it is
not locally entanglement-breaking for any $k$. Moreover, for
$\lambda>0.5567$ the qubit depolarizing channel is not $k$-LEA for 
all $k\geq 3$. We found the following set relations
\be
\sT_{\rm EB}^1=\sT_{\rm 2-LEB}=\sT_{\rm 3-LEB}=\cdots =\sT_{\rm \infty-LEB}\,,\\
\sT_{\rm EB}^1\subset
\sT_{\rm\infty-LEA}\subset\cdots\subset\sT_{\rm 3-LEA}\subsetneq\sT_{\rm 2-LEA}\,,
\ee 
where $\sT_{\rm EB}^1$ is the set of entanglement-breaking channels 
of a single particle. 

The introduced concept of entanglement-annihilating channels 
opens several interesting mathematical and physical questions
related to generic properties of entanglement dynamics.  
For example, we left open the problem of complete characterization of 
entanglement-annihilating channels. For practical purposes, 
it would be also be of interest to find an efficient testing algorithm 
for entanglement-annihilating channels.

\section*{Acknowledgments}
We acknowledge financial support of the European Union project 
HIP FP7-ICT-2007-C-221889, and of the projects
VEGA-2/0092/09, CE SAS QUTE and MSM0021622419. We also thank Daniel Nagaj for
his comments on the manuscript.

\end{document}